\documentclass{article}

\usepackage{verbatim}
\usepackage{graphics}
\usepackage{graphicx}
\usepackage{amsmath, amssymb, amsfonts, amsthm}
\usepackage{url}
\usepackage{setspace}
\newtheorem{thm}{Theorem}[section]

\newtheorem{lem}[thm]{Lemma}

\theoremstyle{definition}
\newtheorem{defn}[thm]{Definition}

\begin{document}

\author{George B. Purdy\footnote{george.purdy@uc.edu. Department Of Computer Science, 814 Rhodes Hall, Cincinnati, OH 45221}\\
Justin W. Smith\footnote{smith5jw@email.uc.edu. Department Of Computer Science, 814 Rhodes Hall, Cincinnati, OH 45221} }

\title{On Finding Ordinary or Monochromatic Intersection Points\footnote{AMS MSC-2000: 51A99, 68U05.  ACM CCS-1998: I.3.5.}}
\maketitle

\begin{abstract}
An algorithm is demonstrated that finds an ordinary intersection in an arrangement of $n$ lines in $\mathbb{R}^2$,
not all parallel and not all passing through a common point, in time $O(n \log{n})$. 
The algorithm is then extended to find an ordinary intersection among an arrangement of hyperplanes in $\mathbb{R}^d$, 
no $d$ passing through a line and not all passing through the same point, again, in time $O(n \log{n})$.

Two additional algorithms are provided that find an ordinary or monochromatic intersection, respectively, in an arrangement of pseudolines in time $O(n^2)$.
\end{abstract}

\section{Introduction}
Over a century ago Sylvester posed the question of whether a set of $n$ non-collinear points necessarily 
determine an ordinary line \cite{Syl93}. (An \emph{ordinary} line is one incident to exactly two points.)  
Although it was thought to be true, no proof was found until the problem was raised again by Erd\H{o}s in the 1930's.
Soon after, it was proven by Gallai and his proof was published in \cite{Ste44}. Hence, it is now called the Sylvester-Gallai Theorem.
(See also \cite{Cox48} for an elegant proof by L. M. Kelly.)

Since Sylvester originally posed his question in 1893,
a variety of related questions have been asked.
One well known variation relates to a two-colored, or bichromatic, set of points.
Ron Graham first asked (around 1965, see \cite{Gru99}) 
whether a bichromatic set of non-collinear points necessarily determines a monochromatic line,
i.e., a line determined by two or more points all of which are the same color.
The first published proof was a few years later by Chakerian \cite{Cha70}.
Earlier than Chakerian (and referenced in his paper), 
Motzkin and Rabin had proofs of this result in its dual form.
(Motzkin's proof was published in \cite{EP95}\footnote{Gr\"{u}nbaum states in \cite{Gru99} 
that the proof published in \cite{EP95} is actually due to Motzkin, although the text attributes it to S. K. Stein.}.)
This theorem is now commonly called the Motzkin-Rabin Theorem.

The algorithms in this article deal with arrangements of hyperplanes (i.e., $(d-1)$-flats in $\mathbb{R}^d$) or pseudolines in the euclidean plane.
By duality, the algorithms on hyperplanes can be used, as well, on a point configuration to solve the dual problem.
However, hyperplane arrangements are more general than a dual of points,
e.g., any two points determine a line, but two parallel lines do not determine an intersection point.
Thus, some problem instances (i.e., those involving parallel hyperplanes) can only be solved by algorithms that work in the domain of hyperplane arrangements.
 
The first algorithm presented finds an ordinary intersection point in an arrangement of lines (some possibly parallel) in $\mathbb{R}^2$
in time $O(n \log{n})$.  This algorithm will subsequently be used to solve the same problem for hyperplanes in $\mathbb{R}^d$.

\section{Ordinary Points in an Arrangement of Lines in $\mathbb{R}^2$}
\subsection{Existence of Ordinary Intersection Points}

Dirac conjectured in 1951 that in any set of $2n$ points, there exist $n$ ordinary lines \cite{Dir51}.
The best known lower bound is $\frac{6n}{13}$ ordinary lines in a set of $n$ points, found by Csima and Sawyer in \cite{CS93}.
This improved upon the Kelly and Moser result of $\frac{3n}{7}$ \cite{KM58}.

Since an ordinary line always exists among a set of non-collinear points, an obvious 
question within computational geometry is how to find one. A naive method would potentially take time $O(n^3)$ by considering 
for each point pair whether a third point is collinear. Mukhopadhyay et al. improved this by finding an algorithm
that finds an ordinary line among a set of points in time $O(n \log{n})$ \cite{MAH97}. 
A similar, but simplified, algorithm was demonstrated several years later by Mukhopadhyay and Green \cite{MG07}.

In \cite{Len04}, Lenchner considers the ``sharp dual'' of this problem, 
i.e., ordinary intersections determined by an arrangement of lines in $\mathbb{R}^2$, 
not all parallel and not all passing through a common point.
(Since the ``sharp dual'' is more general than the ``dual'', the Csima and Sawyer result does not apply.)
Lenchner first proved that ordinary intersections occur in such an arrangement, and in fact, 
that there must exist at least $\frac{5n}{39}$ such points among $n$ lines.
He later improved this original result to $\frac{2n-3}{7}$ among $n \geqslant 7$ lines\cite{Len07}.

In the conclusion of \cite{Len04}, Lenchner asks whether an algorithm exists that 
can find an ordinary intersection in such an arrangement in time $o(n^2)$. 
The following algorithm performs in time $O(n \log{n})$.

\subsection{Locating an Ordinary Intersection}
\begin{defn}
Let $L_0$, $L_1$, and $L_2$ be any three lines of $\mathcal{A}$ that intersect at three distinct points. 
Label the lines such that $L_0$ and $L_1$ intersect at point $P$, which is to the left of the intersection $Q$ of lines $L_0$ and $L_1$.
Points $P$ and $Q$ are \emph{consecutive points} (on $L_0$) if no line intersects $L_0$ on the open interval $(P, Q)$.
Furthermore, if $P$ and $Q$ are consecutive points, $L_1$ and $L_2$ are \emph{consecutive lines} (with respect to $L_0$) 
if $L_1$ is the ``rightmost'' line through $P$ and $L_2$ is the ``leftmost'' line through $Q$.  
In other words, there is no line through $P$ intersecting $L_2$ at a point closer to $Q$ than $L_1 \cap L_2$, 
and there is no line through $Q$ closer to $P$ than $L_1 \cap L_2$.
\end{defn}

\begin{lem}
\label{lem:closest_intsct}
Suppose line $L_0$ contains no ordinary points. 
Let $X$ be the closest intersection point above line $L_0$ incident to at least two lines not parallel to $L_0$.
Then, $X$ must be the intersection of two consecutive lines through two consecutive points $P_i$ and $P_{i+1}$.
\end{lem}
\begin{proof}
Suppose $X$ is the intersection of lines, $L_1$ and $L_2$, through consecutive points, $P_i$ and $P_{i+1}$, but the lines are not consecutive.
Thus, there exists a line through either $P$ or $Q$ that intersects either $L_1$ and $L_2$ at a point closer than $X$, i.e., a contradiction

So, suppose there are three intersection points, $P$, $Q$, and $R$ on $L_0$ in that order from left to right, 
and $X$ is the intersection of a line through $P$ and a line through $R$.  Then there is another line through $Q$ that intersects one side
of the triangle $\triangle XPR$, interior to the side $PX$ or $RX$.  
Either way, there is an intersection point $S$ that is lower than $X$, i.e, a contradiction.
\end{proof}

\begin{figure}
\centering
\includegraphics[height=8cm,width=8cm]{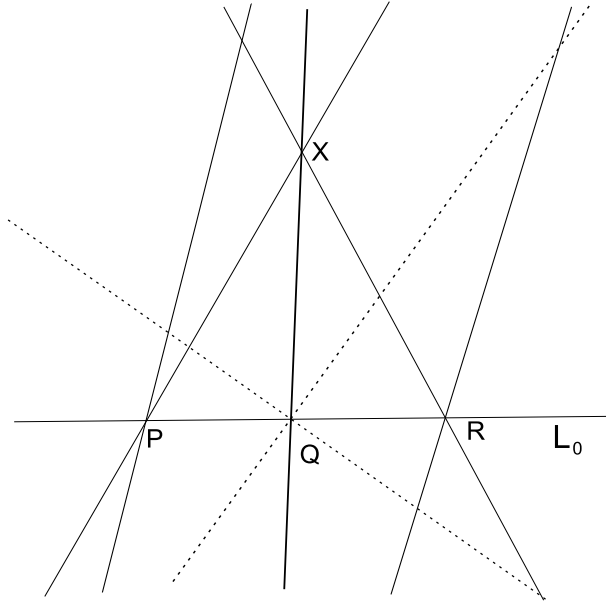}
\caption{A line passing through (non-ordinary) point $Q$ will determine an intersection closer to $L_0$ than $X$.}
\label{fig:closest_intsct}
\end{figure}

\begin{lem}
\label{lem:closest_is_ord}
Suppose line $L_0$ contains no ordinary points, and $X$ is the closest intersection point above $L_0$.
Then $X$ is either an ordinary point, or it has a line $M$ through it parallel to $L_0$.
\end{lem}
\begin{proof}
Let $P$ and $Q$ be points in order from left to right on $L_0$ that contain the lines forming $X$.  
(By Lemma \ref{lem:closest_intsct}, $P$ and $Q$ are consecutive points and $X$ is formed by consecutive lines through these.)
Suppose that there is a third line $\overline{XR}$ through $X$, where $R$ is another point on $L_0$.  
Without loss of generality let's suppose that $R$ is to the right of $Q$.  
There is another line through $Q$ that intersects either segment $PX$ or segment $RX$, 
and either way $X$ is not the lowest point, i.e., a contradiction.
See Figure \ref{fig:closest_intsct}.
\end{proof}

Together, Lemmas \ref{lem:closest_intsct} and \ref{lem:closest_is_ord} can be used to prove the dual form of the Sylvester-Gallai Theorem.
(A configuration of points can always be dualized such that no two lines are parallel.)
From the algorithm below, one can also see proof of its ``sharp dual'' form.

\subsection{Algorithm to Find an Ordinary Point in Time $O(n \log{n})$ }
\label{sub:alg-lines}
\begin{thm}
Given an arrangement of $n$ lines in $\mathbb{R}^2$, not all parallel and not all passing through the same point, 
there exists an algorithm to find a ordinary intersection point in time $O(n \log{n})$.
\end{thm}
\begin{proof}
Let $n$ be the number of lines in arrangement $\mathcal{A}$.
Let $L_0$, $L_1$, and $L_2$ be any three lines of $\mathcal{A}$ that intersect at three distinct points.
If no such lines exist then the arrangement consists of only two families of parallel lines 
and every intersection is ordinary, so suppose this not to be the case. 
\emph{Time to find $L_0$, $L_1$, and $L_2$: $O(n)$.}

Consider $L_0$ to be horizontal, and $L_1$ and $L_2$ to be intersecting ``above'' $L_0$.
Find all of the intersection points on $L_0$. Label them from left to right $P_1$, $P_2$, \ldots, $P_m$.
\emph{Time to sort them, collecting potentially multiple lines into each $P_k$: $O(n \log{n})$.}

If any $P_k$ is ordinary the algorithm is done, so suppose that none are.
 
Find the ``leftmost'' and ``rightmost'' line through each $P_k$, i.e., find the pairs of consecutive lines.  \emph{Time: $O(deg(P_1)+deg(P_2)+\ldots+deg(P_m))=O(n)$}.

Let $X$ be the lowest intersection point above $L_0$, which by Lemma \ref{lem:closest_intsct} 
must be the intersection of consecutive lines through consecutive bundles $P_k$ and $P_{k+1}$.
\emph{Time: $O(n)$}.

\begin{figure}
\centering
\includegraphics[height=7cm,width=8cm]{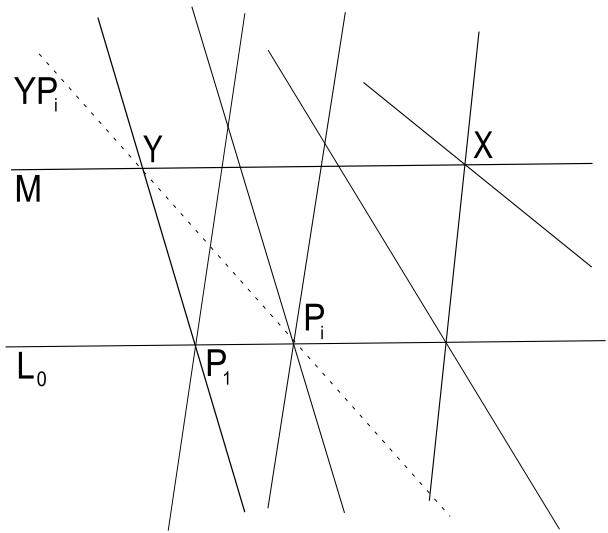}
\caption{If $X$ is the lowest intersection determined by lines not parallel to $L_0$, then $Y$ must be ordinary.}
\label{fig:par_through}
\end{figure}

Determine whether there is a line $M$ parallel to $L_0$ that passes through $X$, and if so, then find $M$. \emph{Time: $O(n)$.}
If there is no such line $M$ then, by Lemma \ref{lem:closest_is_ord}, $X$ is an ordinary point.

Otherwise, suppose that $M$ exists. Let $Y$ be the intersection of $M$ with the leftmost line from the leftmost bundle $P_1$. 

Assume $Y$ is not an ordinary point. Then, there exists a point $P_k$, $k > 1$, such that $YP_k$ is a line of the arrangement, and there is another line
through $P_1$ that intersects $YP_k$ in its interior at a point lower than $Y$ and therefore lower than $X$, i.e., a contradiction.  Hence $Y$ is 
an ordinary point. (See Figure \ref{fig:par_through}.) \emph{Time to find $Y$: $O(1)$.}
\end{proof}

\section{Ordinary Points in an Arrangement of Hyperplanes in $\mathbb{R}^d$}
\subsection{Duality}
Given $n$ points in $\mathbb{R}^3$, no three on a line and not all on a plane, does there necessarily exist a three-point plane?
Recently, the present authors proved that, under the same hypothesis, 
there must exist at least $\frac{4}{13}\binom{n}{2}$ such planes \cite{PS09}.
Without the assumption that no three points are collinear, 
Bonnice and Kelly showed that there must exist at least $\frac{3n}{11}$ ordinary planes in $3$-space \cite{BK71}.
The existence of such a plane also follows from Hansen's result on the existence of 
\emph{ordinary} hyperplanes in $d$-dimensional projective space \cite{Han65}. 
(In $3$-space their existence was proved earlier by Motzkin, but for $4$ or higher it's due to Hansen.)
An ordinary hyperplane in $d$-space is one in which all but one of the points lie on a $(d-2)$-flat.

The algorithm presented in Section \ref{sub:alg-hyperplanes} solves the following problem.
Find the intersection point of exactly $d$ hyperplanes in $\mathbb{R}^d$ (i.e., an ordinary intersection)
in an arrangment of hyperplanes, not all parallel, not all passing through the same point, and no $d$ passing through a line.
This problem on hyperplanes is the ``sharp dual'' of a problem on points. That is, given a set of $n$ points in $\mathbb{R}^d$, 
no $d$ on a $(d-2)$-flat and not all on a hyperplane, find a hyperplane determined by exactly $d$ points.
(It follows from this hypothesis that no $k$ points lie on a $(k-2)$-flat for $3\leqslant k\leqslant d$.)

For the convenience of the reader, we provide the following correspondences between flats and their duals:
\begin{itemize}
\item Hyperplanes $\longleftrightarrow$ Points
\item $(d-2)$-flats $\longleftrightarrow$ Lines
\item $k$-flats $\longleftrightarrow$ $(d-k-1)$-flats
\end{itemize}

\subsection{Algorithm to Find an Ordinary Point in Time $O(n \log{n})$}
\label{sub:alg-hyperplanes}

Before we claim to have an algorithm to find an ordinary intersection, 
we must first be sure that a given set of hyperplanes determines an intersection point.
Let $h_i^{\perp}$ be a normal vector to the hyperplane $h_i$.
The following lemma shows the necessary and sufficient conditions for a general intersection point to exist.
\begin{lem}
\label{lem:normals-span}
Let $H=\{h_0, h_1, \ldots, h_{d-1} \}$ be a set of $d$ hyperplanes in $\mathbb{R}^d$.
The hyperplanes of $H$ determine an intersection point if and only if their normals form a basis for $\mathbb{R}^d$,
i.e., $span(h_0^{\perp}, h_1^{\perp}, \ldots, h_{d-1}^{\perp} \}) = \mathbb{R}^d$.
\end{lem}
\begin{proof}
This follows from the observation that the space orthogonal to the intersection of hyperplanes 
is the span of the hyperplanes' normals.
That is, given a set of $k$ hyperplanes $H'=\{h'_0, h'_1, \ldots, h'_{k-1} \}$, then
$(h'_0 \cap h'_1 \cap \ldots \cap h'_{k-1})^{\perp} = span(h_0^{'\perp}, h_1^{'\perp},\ldots,h_{k-1}^{'\perp})$.
(Note that the sum of the dimension of a space and the dimension of its orthogonal space in $\mathbb{R}^d$ is always $d$).
\end{proof}

Given a set of $k$ vectors one can find a maximal linearly independent subset in time $O(k^2)$ using a method such as row reduction on a matrix.
By letting $k=n$, we could determine such a subset (i.e., a basis for the span) of $n$ vectors, for any $n$, in time $O(n^2)$.
However, one can do better.
\begin{lem}
\label{lem:basis-n-vectors}
A maximal linearly independent subset from a set of $n$ vectors can be found in time $O(n)$.
\end{lem}
\begin{proof}
We assume the dimension $d$ to be constant, 
and thus, a maximal linearly independent set of $d$ vectors can be determined in constant time.

Let $S = \{v_0, v_1, \ldots, v_{n-1}\}$ be the set of $n$ vectors from which a maximal linearly independent set must be found.
Let $M$ be a matrix of dimension $d \times d$ with rows initialized to the vectors $v_0, v_1, \ldots, v_{d-1}$.
Use row reduction, tracking the vector corresponding to each row (e.g., updating as needed upon a row exchange),
to determine a linearly independent subset of these vectors. \emph{Time $O(1)$.} 

If these vectors span $\mathbb{R}^d$ then the algorithm may terminate.
Otherwise, discard the zero rows from the row reduced $M$,
replacing them with vectors from $S$ while still maintaining a correspondence between rows and vectors. \emph{Time $O(1)$.}

Begin the next iteration by performing a row reduction on $M$ to find a linearly independent set.
(Note that one needs only to reduce the rows of $M$ that were most recently added.)
After the row reduction on each iteration, zero rows are replaced with vectors from $S$.
The iterating continues until either $S$ is exhausted or a linearly independent set of size $d$ is found.
This repeats at most $n-d$ times. \emph{Time $O(n)$.}

The resulting set of vectors forms a maximal linearly independent set.
\end{proof}

The following lemma demonstrates the strength of the hypothesis needed by our algorithm.
\begin{lem}
\label{lem:dk-hyperplanes}
Given that no $d$ hyperplanes of a set in $\mathbb{R}^d$ pass through the same line,
no $k$ hyperplanes of that set contain the same $(d-k+1)$-flat for $3\leqslant k\leqslant d$.
\end{lem}
\begin{proof}
Suppose hyperplanes $h_0, h_1, \ldots, h_{k-1}$ all contain a $(d-k+1)$-flat,
then $dim(h_0 \cap h_1 \cap \ldots \cap h_{k-1}) \geqslant d-k+1$.
Thus, by intersecting with an additional $(d-k)$ hyperplanes, for a total of $d$ hyperplanes, 
(assuming no two are parallel) the resulting flat will have dimension at least $(d-k+1)-(d-k)= 1$,
and $dim(h_0 \cap h_1 \cap \ldots \cap h_{d-1}) \geqslant 1$, i.e., a contradiction.
(If any two hyperplanes in the intersection are parallel, the result is an empty set.)
\end{proof}

We will assume that the intersection of two flats can be found in constant time.

\begin{thm}
Given an arrangement of $n$ hyperplanes in $\mathbb{R}^d$, not all parallel, not all passing through the same point, and no $d$ passing through a line,
there exists an algorithm to find an ordinary intersection point, or determine that none exists, in time $O(n \log{n})$.
\end{thm}

\begin{proof}
Let $H$ be the set of $n$ hyperplanes, $\{h_0, h_1, \ldots, h_{n-1}\}$, in $\mathbb{R}^d$, $n \geqslant d$, 
not all parallel, no $d$ passing through a line and not all passing through the same point.

For each hyperplane, compute its normalized normal vector. 
The first non-zero coordinate of each normal vector should be positive (replacing the vector by its negative if necessary),
so that two hyperplanes are parallel if and only if their normal vectors are the same.
\emph{Time: $O(n)$}.

Sort the hyperplanes, lexicographically, by their normals into $k$ families of parallel hyperplanes.
That is, let $H^{(0)}$, $H^{(1)}$, \ldots, $H^{(k-1)}$ each be a set of hyperplanes
such that for any two $h \in H^{(i)}$, $h' \in H^{(j)}$, $h$ and $h'$ are parallel if and only if $i = j$.
Since parallelism is an equivalence relation, this forms a partition on the set $H$.
So, obviously, $|H^{(0)}| + |H^{(1)}| + \ldots + |H^{(k-1)}| = |H| = n$. 
Let $h^{(j)}_i$, for $0 \leqslant i < |H^{(j)}|$, be the members of the set $H^{(j)}$.
\emph{Time: $O(n \log{n})$}.

For each set of hyperplanes, $H^{(i)}$, there is a distinct normal.
Use the algorithm described in Lemma \ref{lem:basis-n-vectors} to find a maximal linearly independent set from these normal vectors,
tracking for each normal the associated hyperplane family.
If the maximal linearly independent set does not span $\mathbb{R}^d$, then by Lemma \ref{lem:normals-span} there will exist no intersection point,
and the algorithm is finished
\emph{Time: $O(n)$.}

From now on we assume that an intersection exists, and therefore the number of hyperplane families, $k$, is at least $d$.

Let $M$ be the plane formed by intersecting a member from each of the first $d-2$ sets of hyperplanes 
from the $d$ sets whose normals formed the basis in the previous step.
Without loss of generality, we will assume that these $d-2$ families are $H^{0}, H^{1}, \ldots, H^{d-3}$.
Thus, $M= h^{(0)}_0 \cap h^{(1)}_0 \cap \ldots \cap h^{(d-3)}_0$.
By Lemma \ref{lem:dk-hyperplanes}, $dim(M)=2$. \emph{Time to determine $M$: $O(n)$.}

Let $L$ be the lines formed by intersecting each of the remaining hyperplane families with $M$, 
i.e., $L = \{l_i = M \cap h_i : h_i \in \left( H^{(d-2)} \cup H^{(d-1)}\cup \ldots \cup H^{(k-1)} \right) \}$, 
where each $l_i$ is a line.
(Note that the hyperplanes used to form lines on $M$ all intersect 
$M$ since they are not parallel to any of the hyperplanes used in the construction of $M$.)
\emph{Time to determine the set $L$: $O(n)$.}

Consider the following cases.

\item{Case 1: The lines of $L$ are all parallel on $M$.}

Consider these hyperplanes in projective space.  There exists a point on the hyperplane at infinity, $p_{\infty}$, that is incident to all lines of $L$,
and thus, incident to all hyperplanes in $\{H^{(d-2)},  H^{(d-1)}, \ldots, H^{(k-1)}\}$.
Furthermore, $p_{\infty}$ is incident to all hyperplanes in $\{H^{(0)}, H^{(1)}, \ldots, H^{(d-3)}\}$, since $M$ and all of its constituent hyperplanes, 
and thus, the parallel hyperplane families, pass through $p_{\infty}$.
Since the hyperplanes of $H$ all share a common point at infinity, there will be no finite intersection point.
However, since we showed that the hyperplane normals spanned $\mathbb{R}^d$, this case is not possible.

\item{Case 2: The lines of $L$ are all concurrent on $M$, and there are at least three of them.}

For this case, we know that $|H^{(d-2)}| = |H^{(d-1)}|= \ldots = |H^{(k-1)}| = 1$.
Let $p$ be the point on $M$ at which the lines of $L$ all cross. \emph{Time to determine if all concurrent: $O(n)$.}

Since the lines of $L$ all pass through $p$, we will construct another plane $M'$ (parallel to $M$) 
by using one alternative member from one of the first $d-2$ parallel families. 

If no alternative member exists, then all hyperplane families contain just one member. 
Thus, all hyperplanes pass through point $p$, in violation of the hypothesis. 
In this case, no ordinary point exists and the algorithm terminates.

Without loss of generality, assume $h^{(0)}_1$ (i.e., a second member from the set $H^{(0)}$) is the alternative member used to construct $M'$.
That is, let $M'= h^{(0)}_1 \cap h^{(1)}_0 \cap \ldots \cap h^{(d-3)}_0$.
\emph{Time to construct $M'$: $O(n)$.}

Construct the set $L'$ in an analogous manner to the construction of $L$, 
i.e., let $L' = \{l'_i = M' \cap h_i : h_i \in \left( H^{(d-2)} \cup H^{(d-1)}\cup \ldots \cup H^{(k-1)} \right) \}$.
The lines of $L'$ cannot all be parallel by the same argument used in Case 1.

Assume the lines of $L'$ are again all concurrent at a finite point $p'$.
Then the line determined by $p$ and $p'$, i.e. $\overline{pp'}$, 
is contained in the $d-3$ hyperplanes used to construct both $M$ and $M'$, 
which excludes the two hyperplanes used from the set $H^{(0)}$. 
(An intersection of hyperplanes containing two distinct points, also contains the line that connects them.)
That is, $\overline{pp'}$ is contained in $h^{(1)}_0,  h^{(2)}_0, \ldots, h^{(d-3)}_0$.
The line $\overline{pp'}$ is also contained in the three or more hyperplanes that formed the lines of $L$ and $L'$.
Altogether, there must be at least $(d-3)+3=d$ hyperplanes containing the line $\overline{pp'}$, in violation of our hypothesis.
\emph{Time to determine if all concurrent on $M'$: $O(n)$.}

Therefore, $M'$ contains an ordinary intersection, and we may proceed to the next case, mutatis mutandis.

\item{Case 3: The lines of $L$ form an ordinary intersection on $M$.}

The ordinary intersection formed by the lines of $L$ can be found using the algorithm given in Section \ref{sub:alg-lines}.
Assume $l_i$ and $l_j$ form an ordinary intersection on $M$.
This point is the intersection of the hyperplanes $h^{(0)}_0 \cap h^{(1)}_0 \cap \ldots \cap h^{(d-3)}_0$ 
and exactly two other hyperplanes (which formed $l_i$ and $l_j$), and thus, at the intersection of $d$ hyperplanes.  
Therefore, we have found an ordinary intersection.
\emph{Time: $O(n \log{n})$.}
\end{proof}

\section{Arrangements of Pseudolines}
\subsection{Ordinary Intersection Points}

Now consider an arrangement of pseudolines, any two of which cross and not all at the same point.
Any such arrangement contains an ordinary intersection, and the best result in this area is 
that of Csima and Sawyer \cite{CS93}, who extended their
$\frac{6n}{13}$ result to also include ordinary intersections among pseudolines in the projective plane.
An elegant proof of the existence of ordinary intersections can be found 
using Euler's formula to find an inequality due to Melchior \cite{Mel41}. See Felsner's book \cite{Fel04} for 
excellent coverage of this and other results related to Sylvester's Problem.

Arrangements of pseudolines, as discussed in this paper, have certain properties that are assumed:
\begin{itemize}
\item Each pseudoline goes off to infinity in both directions.
\item No pseudoline crosses itself.
\item Each pair of pseudolines intersects at exactly one point, and at that point cross.
\item More than two pseudolines may cross at a single point (otherwise the intersection is ordinary).
\item The pseudolines do not all cross at the same point (i.e., there is more than one intersection point).
\end{itemize}
See \cite{Fel04} for a more complete explanation of pseudoline arrangements and their properties.

It is assumed that given a point $P$ and a pseudoline $L$, one can determine whether $P$ lies on $L$ in time $O(1)$.
Therefore, in an arrangement of $n$ pseudolines, the pseudolines that cross $P$ can be determined in time $O(n)$.
It is also assumed that the intersection point of any two pseudolines can be found in time $O(1)$.

Recently a couple of results related to the following algorithm have been published or submitted.
Pretorius and Swanepoel in \cite{PS08} provide proof of a theorem that generalizes both
Sylvester-Gallai and Motzkin-Rabin. Their proof utilizes a sequence of successively smaller triangles
that terminates with finding the desired intersection point.
A similar method is also used by Lenchner in \cite{Len08}.
Note that one might also see similarity between these proofs and Motzkin's as published in \cite{EP95} 
(i.e. they utilize what Motzkin calls ``characteristic triangles'').

The algorithm described below was inspired by the recent proof given by Lenchner in \cite{Len08}.
This algorithm can be used to find ordinary point in an arrangement of lines, and by duality
an ordinary line determined by a set of points. 
We must also mention that a $O(n^2)$ algorithm could be obtained by an 
incremental construction of the arrangement that tracks the intersections that are created.
However, such an algorithm would not also prove the existence of an ordinary intersection point.

\subsection{Algorithm to Find an Ordinary Point in Time $O(n^2)$}

\begin{figure}
\centering
\includegraphics[height=7cm,width=8cm]{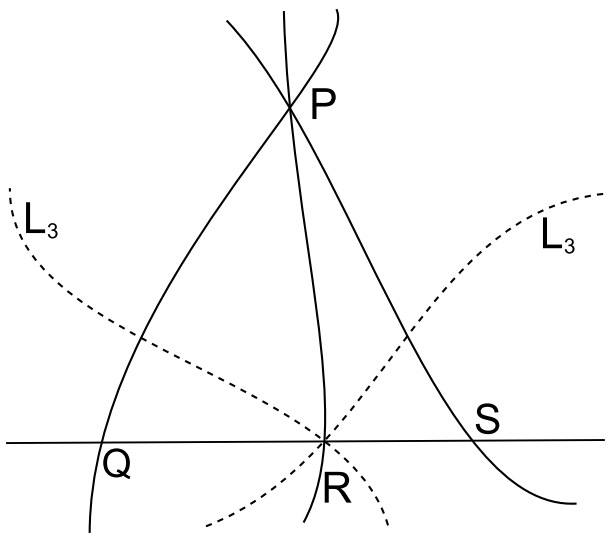}
\caption{If point $R$ is not ordinary, then a third pseudoline $L_3$ must cross either segment $PQ$ or $PS$.}
\label{fig:tri_alg}
\end{figure}

\begin{thm}
An ordinary intersection can be found in an arrangement of pseduolines in time $O(n^2)$.
\end{thm}

\begin{proof}
Let $L_0$, $L_1$ and $L_2$ be any three pseudolines of arrangement $\mathcal{A}$ that intersect at three distinct points.
\emph{Time to find three such lines: O(n).}

Let $P$ be the intersection point of $L_1$ and $L_2$. 
If $P$ is ordinary, then the algorithm is done.
\emph{Time to determine whether P is ordinary: O(n).}

Otherwise, there are at least three pseudolines ($L'_1, L'_2$ and $L'_3$) crossing at $P$.
Consider $L_0$ to be ``horizontal'' and $P$ ``above'' $L_0$.  Let points $Q$, $R$ and $S$ be the points of intersection left to right on $L_0$ 
of the three pseudolines crossing at $P$ (i.e., $L'_1, L'_2$ and $L'_3$).  If $R$ is an ordinary intersection, then the algorithm is done.
\emph{Time to determine whether R is ordinary: O(n).}

Otherwise there is a pseudoline, $L_3$, crossing at $R$ that either crosses the finite segment $QP$ or $PS$.  \emph{Time to determine where $L_3$ crosses: $O(1)$.}

Without loss of generality, assume this pseudoline crossing $R$ also crosses $QP$.  This pseudoline is defined to be the triangle's \emph{dividing line}. 
The configuration will now be reoriented for recursion, letting $R$ be the intersection of $L_3$ 
with pseudoline $QP$, $P$ the previous $R$, $Q$ the previous $P$, and $S$ the previous $Q$.  \emph{Time to reorient the configuration for recursion: $O(1)$.}

The following lemma states that this recursion repeats no more than $n$ times, yielding a time $O(n^2)$ algorithm.
\end{proof}

\begin{lem}
\label{lem:div_line}
No pseudoline is used by the algorithm as a \emph{dividing line} more than once.
\end{lem}
\begin{proof}
Each \emph{dividing line} $L$ crosses the interior of a triangle, dividing it into two parts. 
All subsequent \emph{dividing lines} used by the algorithm must cross the interior of one of those parts, 
of which $L$ lies on the boundary.
\end{proof}

This second algorithm has a potential advantage over our first since it may stop early, 
possibly at the first intersection $P$ (i.e. time $\Omega(n)$). By duality,
this algorithm also can be used to find ordinary lines in an arrangement of points, 
with the same time complexity.

\subsection{Existence of Monochromatic Points in a Bichromatic Arrangement}
In a bichromatic arrangement of pseudolines, any two crossing and not all crossing at the same point,
a monochromatic intersection point always exists, but it might not exist for both colors. 
An arrangement containing monochromatic intersections of only one color is called ``biased'' (see \cite{Gru99}). 
The existence of biased arrangements requires any algorithm in search of 
a monochromatic intersection to consider both colors
(or at least be run twice if limited to a specific color).

The previous algorithm will now be modified to find a monochromatic intersection. 
While Chakerian \cite{Cha70} and others have proven that lines in the real projective plane always determine a monochromatic
intersection (and an argument similar to theirs might be extended to include pseudolines),
the present authors are unaware of a proof that explicitly
extends this result to pseudolines in the euclidean plane. The algorithm below provides such a proof.

In \cite{PS04}, Pretorius and Swanepole
provide an algorithm (i.e. an algorithmic proof) to find a monochromatic line in a bichromatic set of points, apparently in 
time $O(n^2)$ (although the present authors are unaware of a ``worst-case'' instance for their algorithm).
Note that the bichromatic pseudoline problem is more general than that of points, since not every pseudoline arrangement
has a dual.

As with the previous algorithm, it is assumed that given a point $P$ and a pseudoline $L$, 
one can determine whether $P$ lies on $L$ in time $O(1)$.
Therefore, in an arrangement of $n$ pseudolines, the pseudolines that cross $P$ can be determined in time $O(n)$.
It is also assumed that the intersection point of any two pseudolines can be found in time $O(1)$.

\subsection{Algorithm to Find a Monochromatic Intersection in a Bichromatic Arrangement of Pseudolines}

\begin{figure}
\centering
\includegraphics[height=7cm,width=8cm]{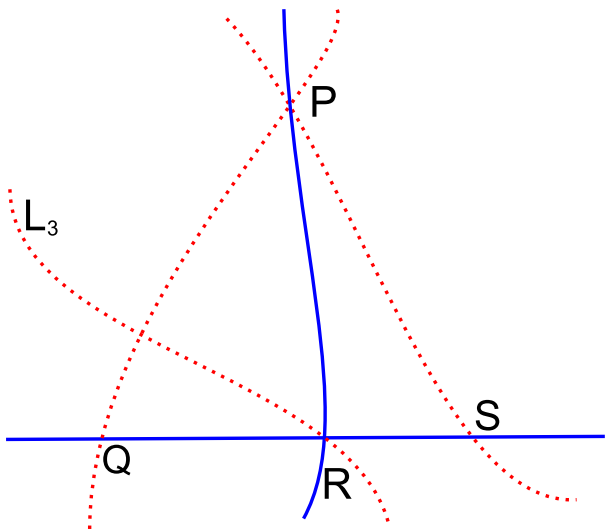}
\caption{If point $R$ is monochromatic, then a third pseudoline $L_3$ with a different color must cross either segment $PQ$ or $PS$.}
\label{fig:bichro_tri_alg}
\end{figure}

\begin{thm}
A monochromatic intersection in a bichromatic arrangement of pseudolines may be found in time $O(n^2)$.
\end{thm}

\begin{proof}
Let $L_0$ be a pseudoline from an arrangement, $\mathcal{A}$, 
containing $n$ pseudolines each colored one of red or blue, 
any two of which cross but not all cross at the same point. 
Consider $L_0$ to be ``horizontal'' and, without loss of generality, assume its color is blue. \emph{Time: $O(1)$.}

Let $Q$ and $S$ be the leftmost and rightmost intersection points on $L_0$.
Assume a red pseudoline crosses at $Q$, another red pseudoline crosses at $S$,
and let $P$ be their intersection point ``above'' $L_0$. 
If this assumption is false (i.e. red pseudolines do not cross both $Q$ and $S$), 
then at least one of $Q$ or $S$ is monochromatic and the algorithm is done. 
\emph{Time to find $Q$ and $S$ and determine whether they are monochromatic: $O(n)$.}

If $P$, the intersection of the red pseudolines crossing $Q$ and $S$, is monochromatic then again, the algorithm is done.
Otherwise, a blue pseudoline $L_2$ crosses $P$ and 
intersects $L_0$ at point $R$, between $Q$ and $S$. See Figure \ref{fig:bichro_tri_alg}.

At this point, the ``setup'' is complete and we begin the first step of a (potentially) 
recursive process to find a monochromatic intersection.

If $R$ it is monochromatic (blue), then the algorithm is done. \emph{Time to determine whether $R$ is monochromatic: $O(n)$.}

Otherwise, a red pseudoline, $L_3$, crosses $R$ and intersects either segment $PQ$ or segment $PS$. Without loss of generality,
assume it crosses $PQ$.  This pseudoline is defined to be the triangle's (i.e. $\triangle PQR$'s) \emph{dividing line}.
The configuration will now be reoriented for recursion, letting $R$ be the intersection of $L_3$ 
with pseudoline $QP$, $P$ the previous $R$, $Q$ the previous $P$, and $S$ the previous $Q$.  \emph{Time to reorient the configuration for recursion: $O(1)$.} 

Note that each step of the recursive process, expects $R$ to possess a different, possibly monochromatic, color.  
So for the first and all other odd numbered steps it would expect ``blue'', and likewise ``red'' for the even.
Again, we refer to Lemma \ref{lem:div_line} to show that this algorithm runs in time $O(n^2)$.
\end{proof}

\section{Conclusion}
It is conjectured that both $O(n \log{n})$ algorithms presented here are within a constant factor of the best upper bound for time.

It would be interesting to know whether an algorithm to find an ordinary intersection in an arrangement
of pseudolines could also perform in time $O(n \log{n})$. Likewise, it would be
interesting to know whether an algorithm to find a monochromatic intersection
in a bichromatic arrangement of pseudolines (or even lines) could perform in time $O(n \log{n})$.

\bibliographystyle{amsplain}
\bibliography{finding_intersections}

\end{document}